\documentclass[
  aps,
  prl,
  reprint,
  amsmath,
  amssymb,
  nofootinbib,
  longbibliography
]{revtex4-2}

\usepackage{bm}
\usepackage{graphicx}
\usepackage{amsthm}
\usepackage{microtype}
\usepackage{pgfplots}
\pgfplotsset{compat=1.18}
\usepackage[colorlinks=true,allcolors=blue]{hyperref}

\newcommand{\cat}{\mathrm{cat}}

\newcommand{\Fid}{F}
\newcommand{\TD}{D}

\newcommand{\eps}{\varepsilon}

\newtheorem{theorem}{Theorem}
\newtheorem*{mainthmrestated}{Theorem~\ref{thm:main}}
\newtheorem{lemma}[theorem]{Lemma}
\newtheorem{proposition}[theorem]{Proposition}
\newtheorem{corollary}[theorem]{Corollary}

\begin{document}

\title{Entanglement Catalysis Beyond the Asymptotic Rate}

\author{Alexander Streltsov}
\email{astrel@ippt.pan.pl}
\affiliation{Institute of Fundamental Technological Research, Polish Academy of Sciences, Pawi\'nskiego 5B, 02-106 Warsaw, Poland}

\begin{abstract}
Quantum entanglement is a fundamental resource in quantum information science, commonly consumed to perform tasks like quantum teleportation and quantum key distribution. Entanglement concentration converts partially entangled pure states into nearly perfect Bell pairs. While its efficiency is well understood in the asymptotic limit, practical protocols must operate with finite systems. In this finite regime, and at the low error levels relevant in practice, performance is penalized by corrections that reduce the achievable yield. We investigate whether this fundamental finite-size penalty can be mitigated by an entanglement catalyst: an auxiliary state that aids the transformation but is returned exactly and uncorrupted. Although such a catalyst cannot increase the leading-order asymptotic yield, we show that it fundamentally alters the finite-size scaling. For every bipartite pure state with nonuniform Schmidt coefficients, we exactly determine the optimal second-order correction to the yield and show that catalysis completely overcomes this finite-size penalty.
\end{abstract}

\maketitle

\emph{Introduction.---} Much of quantum information theory is organized around asymptotic laws: when arbitrarily many identical systems are available, complicated transformation problems often collapse to simple entropic rates. Real protocols, however, never operate in this limit. Their performance is controlled by finite-size corrections \cite{TomamichelHayashi2013,Li2014}, and these corrections can reveal operational advantages that are invisible at leading order. Entanglement concentration provides a natural setting to quantify such finite-size effects, as its asymptotic rate and subleading corrections can be characterized explicitly \cite{HorodeckiEtAl2009,KumagaiHayashi2013,DattaLeditzky2015}.

A central task in quantum information science is to convert one form of entanglement into another using local operations and classical communication (LOCC). Entanglement concentration is a canonical example where this question can be addressed quantitatively \cite{BennettEtAl1996}: \(n\) copies of an entangled pure state are compressed into nearly perfect Bell pairs, which can then be used for teleportation, distributed computation, or other network tasks \cite{BennettEtAl1993,WehnerElkoussHanson2018}. In the ideal many-copy limit the optimal conversion rate is simple: the
entanglement entropy gives the number of Bell pairs per input copy
\cite{BennettEtAl1996}.  Real protocols, however, operate with finite $n$
and tolerate a specified error.  Their yield is therefore not determined by
the entanglement entropy alone: the leading finite-size corrections are of
order $\sqrt n$, and they quantify whether the
attainable yield lies above or below the asymptotic benchmark at a given block
length \cite{DattaLeditzky2015,DattaLeditzkyCorrection2017,KumagaiHayashi2017,ChubbTomamichelKorzekwa2018,KorzekwaChubbTomamichel2019,LipkaBartosikEtAl2024}.

Quantum catalysis adds an auxiliary entangled state that can enable a
transformation while being returned unchanged
\cite{JonathanPlenio1999,DattaKondraMillerStreltsov2023,LipkaBartosikWilmingNg2024}.
Its ability to improve finite-copy entanglement concentration was already
noted by Jonathan and Plenio \cite{JonathanPlenio1999}. Subsequently,
Santra and Malinovsky optimized the probability of converting finitely
many copies of a two-qubit state into one Bell pair using a two-qubit
catalyst \cite{SantraMalinovsky2021}. These finite-size advantages coexist
with an unchanged leading asymptotic singlet-distillation rate for
bipartite pure states \cite{GanardiKondraStreltsov2024}.
This motivates a quantitative question: how much can a catalyst improve
the optimal concentration yield beyond the leading entropy rate?
In this work, we determine this improvement at second order. For every entangled pure source with nonuniform Schmidt
coefficients and every fixed error tolerance
$0<\varepsilon<1/\sqrt{2}$, we show that catalysis reverses the sign
of the optimal $\sqrt n$ correction: a finite-size loss without catalysis
becomes a gain, while the leading entropy rate remains unchanged.

In this work, we focus on exact catalysis~\cite{JonathanPlenio1999}. Under this paradigm, the catalyst is a finite-dimensional pure state that is returned exactly in its initial state without any residual correlations with the main output. Exact return excludes entanglement embezzlement, which becomes possible when arbitrarily large catalysts may be returned only approximately \cite{vanDamHayden2003}. Consequently, any operational advantage must arise purely from the catalyst's assistance during the transformation, rather than from its degradation or from hidden correlations left behind. Constructive links between multiple-copy transformations and catalysts \cite{DuanFengLiYing2005} and finite sufficient conditions for catalytic majorization \cite{ElkoussEtAl2026} provide complementary tools for studying exact catalysis.

By contrast, other recent frameworks relax these stringent requirements to expand the set of possible transformations. In correlated catalysis, the catalyst's reduced state is returned exactly, but correlations with the system are allowed \cite{KondraDattaStreltsov2021,Muller2018,ShiraishiSagawa2021,Wilming2021}. For pure-state entanglement, allowing residual correlations and a small transformation error can simplify state convertibility to basic entropy conditions \cite{KondraDattaStreltsov2021}. For exact mixed-state entanglement dilution, catalysis can lower the asymptotic entanglement cost when correlations between the system and catalyst are allowed \cite{Ganardi2026}. However, forcing these residual system--catalyst correlations or conversion errors to be small can necessitate increasingly resourceful catalysts \cite{RubboliTomamichel2022}. Finite-dimensional catalysts and catalyst-size requirements have been studied in general resource theories \cite{LipkaBartosikKorzekwa2025}. Robustness against noise in the input state is addressed in Ref.~\cite{SonEtAl2026}, while catalytic concentration with mixed inputs, state-preparation and operational errors, and catalyst reuse under noise is studied in Ref.~\cite{SharmaEtAl2026}.

Beyond pure-state entanglement, catalytic effects arise across a broad range of
quantum resource theories.  In quantum coherence, catalysts can enable or
amplify coherence transformations
\cite{TakagiShiraishi2022,KondraGanardiStreltsov2024}, while universal
entanglement catalysts and catalytic assistance for noisy quantum channels
were developed in Ref.~\cite{DattaKondraMillerStreltsov2024}.  Other recent
directions include flexible catalysis
\cite{WeiszStrelchuk2025,AoPhilipStreltsov2026}, correlated catalysis in the
PPT resource theory of entanglement \cite{AoPhilipStreltsovPPT2026}, and
limitations on catalytic entanglement distillation
\cite{LamiRegulaStreltsov2024}.  Catalytic assistance has also been explored
in local state discrimination and quantum data hiding
\cite{HalderPhilipStreltsov2026,PhilipStreltsov2025}.  In quantum
thermodynamics, catalytic thermal
operations give
rise to families of generalized free-energy constraints, while their power
depends sensitively on return accuracy and physical restrictions such as
catalyst dimension or energy
\cite{BrandaoEtAl2015,NgEtAl2015,LipkaBartosikWilmingNg2024,
DattaKondraMillerStreltsov2023,Wilming2022}.

\medskip

\emph{Catalytic transformations of bipartite pure states.---}Let $p=(p_1,\ldots,p_D)$ be a nonuniform full-support probability vector,
interpreted as the squared Schmidt coefficients of a bipartite pure state. Throughout, $\log$ denotes the base-2 logarithm and $\ln$ denotes the natural logarithm.
Define the Shannon entropy $H=H(p)$, which is also the entropy of
entanglement of the associated bipartite pure state, and the varentropy $V=V(p)$ by 
\begin{equation}
 H=-\sum_i p_i\log p_i,
 \qquad
 V=\sum_i p_i[-\log p_i-H]^2.
 \label{eq:HV}
\end{equation}
Let $u_2=(1/2,1/2)$ be the squared Schmidt vector of one Bell pair.  For
probability vectors $x,y$, append zeros to a common dimension and write
$x^\downarrow,y^\downarrow$ for their decreasing rearrangements.  We say that $x$ is majorized by $y$, written $x\prec y$, when
\begin{equation}
 \sum_{i=1}^{k}x_i^\downarrow\le
 \sum_{i=1}^{k}y_i^\downarrow
 \quad(1\le k<\dim x),
 \label{eq:majorization}
\end{equation}
with equality for the total sums.  Let $\lvert\psi_x\rangle$ and
$\lvert\psi_y\rangle$ be bipartite pure states with squared Schmidt vectors
$x$ and $y$, respectively.  Nielsen's theorem states that
$\lvert\psi_x\rangle$ can be converted deterministically by LOCC into
$\lvert\psi_y\rangle$ if and only if $x\prec y$ \cite{Nielsen1999}.

Let $\lvert\chi_r\rangle$ be a finite-dimensional pure auxiliary state with
squared Schmidt vector $r$.  We call $\lvert\chi_r\rangle$ a catalyst for the
conversion $\lvert\psi_x\rangle\to\lvert\psi_y\rangle$ if
\begin{equation}
 \lvert\psi_x\rangle\otimes\lvert\chi_r\rangle
 \xrightarrow{\mathrm{LOCC}}
 \lvert\psi_y\rangle\otimes\lvert\chi_r\rangle,
 \label{eq:state-catalysis}
\end{equation}
so the auxiliary state is returned exactly and uncorrelated with the output.
By Nielsen's theorem, Eq.~\eqref{eq:state-catalysis} is equivalent to
$x\otimes r\prec y\otimes r$.  We therefore write $x\prec_T y$ if there
exists a finite probability vector $r$ with this property.  There exist pairs
with $x\nprec y$ but $x\prec_T y$, so a catalyst can enable a deterministic
pure-state LOCC transformation that is impossible without it
\cite{JonathanPlenio1999}.

For a probability vector $x$, define the Klimesh functions \cite{Klimesh2007}
\begin{equation}
 f_\alpha(x)=
 \begin{cases}
  \displaystyle \log\!\sum_i x_i^\alpha, & \alpha>1,\\[1mm]
  \displaystyle \sum_i x_i\log x_i=-H(x), & \alpha=1,\\[1mm]
  \displaystyle -\log\!\sum_i x_i^\alpha, & 0<\alpha<1,\\[1mm]
  \displaystyle -\sum_i\log x_i, & \alpha=0,\\[1mm]
  \displaystyle \log\!\sum_i x_i^\alpha, & \alpha<0.
 \end{cases}
 \label{eq:Klimesh-functions}
\end{equation}
For $\alpha\le0$, if $x$ has a zero component, we set
$f_\alpha(x)=+\infty$.  For pairs $x,y$ that are not related by a
permutation and do not both contain zero components, the Klimesh--Turgut
theorem \cite{Klimesh2007,Turgut2007} states that
\begin{equation}
 x\prec_T y
 \quad\Longleftrightarrow\quad
 f_\alpha(x)<f_\alpha(y)\ \text{for every }\alpha\in\mathbb R.
 \label{eq:Klimesh}
\end{equation}

Aubrun and Nechita characterized the $\ell_1$-closure of the set
of probability vectors that can be catalytically converted into
a fixed target $y$~\cite{AubrunNechita2008}. Specifically, a source
vector $x$ can be approximated arbitrarily closely in $\ell_1$
by finite probability vectors $x'$ satisfying $x'\prec_T y$ if and only if
$\|x\|_\alpha\le\|y\|_\alpha$ for all $\alpha\ge1$, where
$\|x\|_\alpha=(\sum_i x_i^\alpha)^{1/\alpha}$ is the
$\ell_\alpha$ norm.

\medskip

\emph{Catalytic second-order concentration.---}For quantum states
$\rho$ and $\sigma$, let
$\Fid(\rho,\sigma)=\|\sqrt{\rho}\sqrt{\sigma}\|_1^2$ denote the squared fidelity and $\TD(\rho,\sigma)=\tfrac12\|\rho-\sigma\|_1$ the
trace distance, where the trace norm is
$\|A\|_1=\operatorname{Tr}\sqrt{A^\dagger A}$.  For pure states,
$\TD=\sqrt{1-\Fid}$.  If their Schmidt
bases are aligned and their squared Schmidt vectors are $x$ and $y$, these
quantities reduce to
\begin{equation}
 \Fid(x,y)=\left(\sum_i\sqrt{x_i y_i}\right)^2,
 \qquad
 \TD(x,y)=\sqrt{1-\Fid(x,y)}.
 \label{eq:distance}
\end{equation}
   
Let $M_n^{\cat}(\eps)$ be the largest integer $m$ for which probability
vectors $r_n,s_n$ satisfy
\begin{equation}
 p^{\otimes n}\otimes r_n\prec s_n\otimes r_n,
 \qquad
 \TD(s_n,u_2^{\otimes m})\le\eps .
 \label{eq:conversion}
\end{equation}
We denote by $M_n(\eps)$ the corresponding unassisted optimum, obtained by
requiring $p^{\otimes n}\prec s_n$ with the same error constraint.
The theorem below determines the optimal fixed-error asymptotics through the
leading finite-size correction.

\begin{theorem}
\label{thm:main}
For every nonuniform full-support $p$ and every fixed $0<\eps<1$,
\begin{equation}
 M_n^{\cat}(\eps)=nH+b_{\cat}(\eps)\sqrt n+o(\sqrt n),
 \label{eq:second-order-expansion}
\end{equation}
where the optimal coefficient is
\begin{equation}
 b_{\cat}(\eps)=
 \lim_{n\to\infty}\frac{M_n^{\cat}(\eps)-nH}{\sqrt n}
 =\sqrt{2(\ln 2)V\log\!\frac{1}{1-\eps^2}}.
 \label{eq:main}
\end{equation}
\end{theorem}

Without catalysis, optimal approximate pure-state conversion is characterized
by Ref.~\cite{VidalJonathanNielsen2000}, and the established second-order
expansion of $M_n(\eps)$ has coefficient
\cite{DattaLeditzky2015,DattaLeditzkyCorrection2017,KumagaiHayashi2017,ChubbTomamichelKorzekwa2018}
\begin{equation}
 b_0(\eps)=\sqrt V\,\Phi^{-1}(\eps^2),
 \label{eq:b0}
\end{equation}
where $\Phi$ is the standard normal cumulative distribution function.
Consequently,
\begin{equation}
 0<\eps<\frac1{\sqrt2}
 \quad\Longrightarrow\quad b_0(\eps)<0<b_{\cat}(\eps).
 \label{eq:universal-sign}
\end{equation}

\begin{figure}[t]
 \centering
 \includegraphics[width=\columnwidth]{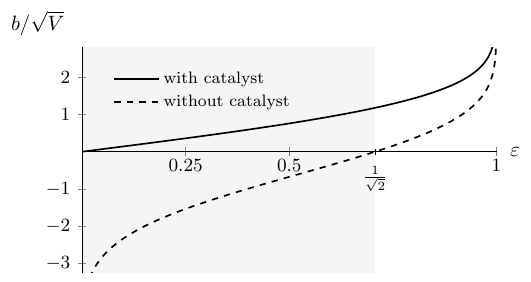}
 \caption{Universal fixed-error second-order coefficients normalized by
 $\sqrt V$. With catalysis we have
 $b_{\cat}/\sqrt V=\sqrt{2\ln[1/(1-\eps^2)]}$ (solid), while the non-catalytic case gives $b_0/\sqrt V=\Phi^{-1}(\eps^2)$ (dashed). The shaded interval
 $0<\eps<1/\sqrt2$ shows the sign-reversal regime.}
 \label{fig:coefficients}
\end{figure}

The theorem describes a qualitative change in finite-block performance.  In
the sign-reversal range, ordinary LOCC must produce fewer than $nH$ Bell pairs
by an amount proportional to $\sqrt n$, whereas an exact catalyst permits more
than $nH$ Bell pairs by the same scaling. The separation becomes especially
pronounced in the limit $\varepsilon \rightarrow 0$, where we have $b_{\cat}\to0$ while $b_0\to-\infty$.
Here $n\to\infty$ is taken first at fixed $\eps>0$; this does not describe
the exactly zero-error task, whose deterministic asymptotic entanglement
exchange rate is unchanged by catalysts \cite{DuanFengJiYing2005}.
The two universal normalized coefficient curves, including the sign-reversal
interval, are shown in Fig.~\ref{fig:coefficients}.

We present the main idea of the proof below; details are given in the Supplemental
Material.  We start with the achievability part. The construction separates the two
tasks that an approximate concentration output must satisfy.  Most of the
probability mass should look like the uniform spectrum of $m_n$ Bell pairs,
while enough nonuniformity must remain to satisfy every catalytic-majorization
constraint.  We realize this with an output state of the form
\begin{equation}
 s_n=\left(
 \eta,
 \underbrace{\frac{1-\eta}{N_n-1},\ldots,
 \frac{1-\eta}{N_n-1}}_{N_n-1\ \mathrm{entries}}
 \right).
 \label{eq:spiked-flat}
\end{equation}
In the following, we fix $\eta=\eps^2$. Moreover we can choose any
$0<\beta<B(\eps)$, where
\begin{equation}
 B(\eps):=\sqrt{2(\ln 2)V[-\log(1-\eta)]}.
 \label{eq:B}
\end{equation}
Set $m_n=\lfloor nH+\beta\sqrt n\rfloor$, $N_n=2^{m_n}$, and let
$u_{N_n}=(1/N_n,\ldots,1/N_n)=u_2^{\otimes m_n}$.  We use
$h_2(q)=-q\log q-(1-q)\log(1-q)$ for the binary entropy.  As we prove in the
Supplemental Material,
\begin{align}
 \Fid(s_n,u_{N_n})
 &=\frac{[\sqrt\eta+\sqrt{(N_n-1)(1-\eta)}]^2}{N_n}
 >1-\eta,\label{eq:fidelity}\\
 H(s_n)&=h_2(\eta)+(1-\eta)\log(N_n-1)<nH
 \label{eq:entropy}
\end{align}
for all sufficiently large $n$.  

The remaining issue is the existence of a catalyst.  For this, we need to
verify that the Klimesh--Turgut inequalities in Eq.~\eqref{eq:Klimesh} hold
in this case.  Because the padded output
has zeros while the source has full support, the nonpositive-order conditions
are automatic.  More explicitly, nonuniformity of $p$ gives $H<\log D$, while
$\log N_n=nH+O(\sqrt n)$, and hence $N_n<D^n$ for all sufficiently large
$n$.  After padding $s_n$ to dimension $D^n$, it therefore contains zeros.
Thus the Klimesh functions $f_\alpha$ of the source are finite for $\alpha\le0$, whereas
those of the padded output diverge; the nonpositive-order inequalities
are automatic, and the support difference also rules out the trivial
permutation case.  

It remains to prove the positive-order conditions in
Eq.~\eqref{eq:Klimesh}, namely
\begin{subequations}\label{eq:KT-special-main}
\begin{align}
 \sum_i s_{n,i}^\alpha
 &>\left(\sum_i p_i^\alpha\right)^n, &&\alpha>1,\\
 H(s_n)&<nH,\\
 \sum_i s_{n,i}^\alpha
 &<\left(\sum_i p_i^\alpha\right)^n, &&0<\alpha<1.
\end{align}
\end{subequations}
The middle inequality follows directly from Eq.~\eqref{eq:entropy}.
The Supplemental Material proves both power-sum statements.  Interestingly,
the upper bound $\beta<B(\eps)$ is needed only for the case $0<\alpha<1$;
the $\alpha>1$ inequalities hold for every fixed real $\beta$.  The proof
is by contradiction and exhausts all possible asymptotic behaviors of a
hypothetical violating sequence of orders. The Klimesh--Turgut theorem then proves the existence of a
catalyst for every sufficiently large $n$.

The converse starts from an arbitrary feasible protocol; it does not
assume the output form in Eq.~\eqref{eq:spiked-flat}.  Consider any
feasible output $s_n$ with target
support $N_n=2^{m_n}$ and any catalyst $r_n$.  Fix $t>0$ and choose
\begin{equation}
 \alpha_n=1-\frac{t}{\sqrt n}.
 \label{eq:converse-alpha}
\end{equation}
For all sufficiently large $n$, $1/2<\alpha_n<1$.  For
$0<\alpha_n<1$, the power sum $z\mapsto\sum_k z_k^{\alpha_n}$ is Schur
concave because $x\mapsto x^{\alpha_n}$ is concave.  Hence the catalytic
majorization relation $p^{\otimes n}\otimes r_n\prec s_n\otimes r_n$ gives
\begin{equation}
 \sum_{i,j}(p_i^{\otimes n}r_{n,j})^{\alpha_n}
 \ge \sum_{i,j}(s_{n,i}r_{n,j})^{\alpha_n}.
 \label{eq:converse-product}
\end{equation}
Both sides factorize and the positive factor $\sum_j r_{n,j}^{\alpha_n}$ cancels exactly, leaving
\begin{equation}
 \left(\sum_i p_i^{\alpha_n}\right)^n
 \ge \sum_i s_{n,i}^{\alpha_n}.
 \label{eq:converse-power}
\end{equation}

It remains to lower-bound the output power sum using the target fidelity.
Lemma~\ref{supp:lem:Jensen} of the Supplemental Material states that if
$\Fid(s,u_N)\ge1-\eps^2$, then for every $\alpha\ge1/2$,
\begin{equation}
 \sum_i s_i^{\alpha}
 \ge N^{1-\alpha}(1-\eps^2)^{\alpha}.
 \label{eq:converse-fidelity}
\end{equation}
Applying this lemma to $s=s_n$, $N=N_n$, and $\alpha=\alpha_n$ supplies the
required bound.  Combining Eqs.~\eqref{eq:converse-power} and
\eqref{eq:converse-fidelity}, and defining
$g(\alpha)=\log\sum_i p_i^\alpha$, gives
\begin{equation}
 n g(\alpha_n)
 \ge (1-\alpha_n)m_n+\alpha_n\log(1-\eps^2).
 \label{eq:converse-log}
\end{equation}
The expansion about order one is
\begin{equation}
 g(1+x)=-Hx+\frac{(\ln2)V}{2}x^2+O(|x|^3).
 \label{eq:converse-cumulant}
\end{equation}
Write
$m_n=nH+b_n\sqrt n$.  Substituting
$\alpha_n=1-t/\sqrt n$ into Eq.~\eqref{eq:converse-log} and using
Eq.~\eqref{eq:converse-cumulant} leads to
\begin{equation}
 b_n\le \frac{(\ln 2)Vt}{2}
 +\frac{-\log(1-\eps^2)}{t}+o(1).
 \label{eq:converse-bn}
\end{equation}
Therefore every feasible sequence satisfies
\begin{equation}
 \limsup_{n\to\infty}b_n
 \le \frac{(\ln 2)Vt}{2}
 +\frac{-\log(1-\eps^2)}{t}
 \qquad(t>0).
 \label{eq:converse-t}
\end{equation}
The right-hand side is minimized at
\begin{equation}
 t_*=\sqrt{\frac{2[-\log(1-\eps^2)]}{(\ln2)V}},
 \label{eq:converse-tstar}
\end{equation}
where its value is
$\sqrt{2(\ln2)V[-\log(1-\eps^2)]}=B(\eps)$.  Thus
$\limsup b_n\le B(\eps)$ for every feasible sequence.  Together with the
achievability construction for every $0<\beta<B(\eps)$, this proves the
converse and hence Eq.~\eqref{eq:main}.

After normalization by $\sqrt V$, all nonuniform sources share the same two
coefficient curves.  Entropy fixes only the extensive rate, whereas
varentropy sets the scale of the second-order finite-block correction.  The boundary
$\eps=1/\sqrt2$ is universal because the unassisted coefficient
$b_0(\eps)=\sqrt V\,\Phi^{-1}(\eps^2)$ changes sign precisely when
$\eps^2=1/2$.  By contrast, $b_{\cat}(\eps)>0$ for every $\eps>0$.
Hence the location of the sign-reversal threshold is independent of the source
$p$; the source-dependent varentropy $V$ only sets the overall scale of the
second-order corrections.

The full-support assumption in Theorem~\ref{thm:main} entails no loss of generality.  If the original Schmidt vector contains zero entries, we restrict the local Hilbert spaces to its Schmidt support.  This leaves the nonzero Schmidt coefficients and the quantities $H$ and $V$ unchanged.  Hence Theorem~\ref{thm:main} applies equally to entangled pure-state sources whose Schmidt vectors are not full support.

\medskip

\emph{Conclusion.---}Our findings demonstrate that entanglement catalysis can fundamentally alter the finite-block laws of quantum state transformations, even when it provides no advantage to the leading-order asymptotic rate. For the concentration of any nonuniform bipartite pure state into nearly perfect Bell pairs, we established the exact second-order finite-size correction under exact catalysis. The result reveals a striking structural property: whereas unassisted protocols suffer a universal performance penalty (a negative second-order correction) for any error tolerance below $1/\sqrt{2}$, catalysis completely overcomes this penalty. The catalytic correction remains strictly positive across this entire error regime. 

This advantage occurs despite the catalyst being returned exactly in its initial state and completely uncorrelated with the system. It highlights an essential operational difference between asymptotic capabilities and finite-resource performance: while the entanglement entropy $H$ sets the first-order rate, the source varentropy $V$ governs a finite-size scale that can be explicitly manipulated by catalytic assistance. 

Our work opens several avenues for future exploration. While we have proven the existence of a catalyst that secures this second-order improvement, quantifying the physical resources required to construct it---specifically bounding its Schmidt rank and entanglement---remains an important open challenge. Characterizing the $o(\sqrt{n})$ remainder, understanding the vanishing-error limit, and extending these results to mixed-state concentration or more general target state transformations will further illuminate the power of finite-size catalysis. Finally, because exact catalysis appears across various resource theories, including quantum coherence and thermodynamics, investigating analogous finite-size sign reversals in those settings may reveal that catalytic resources are fundamentally more powerful in finite regimes than their asymptotic limits imply.

\textit{Acknowledgments.}---This work was supported by the National Science Centre Poland (Grant No. 2024/55/B/ST2/01590).

\emph{Disclosure on the use of artificial intelligence.} The author used OpenAI's GPT 6 Astra to assist with proof development, language editing, manuscript organization, and bibliographic checking. The author critically evaluated the AI-generated suggestions and independently verified the final proofs, references, and conclusions, and takes full responsibility for the content of the manuscript.

\bibliography{references}

\clearpage
\onecolumngrid
\setcounter{secnumdepth}{2}

\begin{center}
  \textbf{\large Supplemental Material}\\[1ex]
\end{center}

\section{Definitions and main statement}

Throughout, $\log$ denotes the base-2 logarithm and $\ln$ denotes the natural logarithm.  For probability vectors $x,y$, with zeros
appended when necessary, define squared fidelity and pure-state trace distance by
\begin{equation}
 \Fid(x,y)=\left(\sum_i\sqrt{x_i y_i}\right)^2,
 \qquad
 \TD(x,y)=\sqrt{1-\Fid(x,y)}.
 \label{supp:eq:distance}
\end{equation}

Let $p=(p_1,\ldots,p_D)$ be a full-support nonuniform probability vector and
let $u_2=(1/2,1/2)$ be the Schmidt vector of one Bell pair.  Denote by
$H=H(p)$ the Shannon entropy of $p$ (equivalently, the entropy of
entanglement of the associated bipartite pure state) and by $V=V(p)$ its
varentropy, or entropy variance.  With the first-order Bell-pair rate $R=H$,
these quantities are
\begin{equation}
 \begin{aligned}
  H&=H(p)=-\sum_{i=1}^D p_i\log p_i,\\
  V&=V(p)=\sum_{i=1}^D p_i\bigl[-\log p_i-H(p)\bigr]^2.
 \end{aligned}
 \label{supp:eq:HVR}
\end{equation}
For $0<\eps<1$, let $M_n^{\cat}(\eps)$ be the largest integer $m$ for which
there exist probability vectors $r_n,s_n$ satisfying
\begin{equation}
 p^{\otimes n}\otimes r_n\prec s_n\otimes r_n,
 \qquad
 \TD(s_n,u_2^{\otimes m})\le\eps.
 \label{supp:eq:conversion}
\end{equation}
We denote by $M_n(\eps)$ the corresponding unassisted optimum, obtained by
requiring $p^{\otimes n}\prec s_n$ with the same error constraint.
\begin{mainthmrestated}
For every nonuniform full-support $p$ and every fixed
$0<\eps<1$,
\begin{equation}
 M_n^{\cat}(\eps)=nH+b_{\cat}(\eps)\sqrt n+o(\sqrt n),
 \label{supp:eq:second-order-expansion}
\end{equation}
where the optimal coefficient is
\begin{equation}
 b_{\cat}(\eps)
 =\lim_{n\to\infty}\frac{M_n^{\cat}(\eps)-nH}{\sqrt n}
 =\sqrt{2(\ln 2)V\log\!\frac{1}{1-\eps^2}}.
 \label{supp:eq:main-formula}
\end{equation}
\end{mainthmrestated}

For comparison, without catalysis the optimal approximate pure-state conversion
is characterized by Ref.~\cite{VidalJonathanNielsen2000}, and the established
second-order expansion of $M_n(\eps)$ has coefficient
\cite{DattaLeditzky2015,DattaLeditzkyCorrection2017,KumagaiHayashi2017,ChubbTomamichelKorzekwa2018}
\begin{equation}
 b_0(\eps)=\sqrt V\,\Phi^{-1}(\eps^2).
 \label{supp:eq:uncat-formula}
\end{equation}
Here $\Phi$ denotes the cumulative distribution function of the standard normal distribution.

\begin{corollary}
\label{supp:cor:sign-reversal}
For every nonuniform full-support $p$ and every
$0<\eps<1/\sqrt2$,
\begin{equation*}
 b_0(\eps)<0<b_{\cat}(\eps).
\end{equation*}
\end{corollary}

The full-support restriction entails no loss of generality.  If the source
Schmidt vector contains zero entries, we restrict the local Hilbert spaces to
its Schmidt support.  This leaves the nonzero Schmidt coefficients, the LOCC
conversion problem, and the quantities $H$ and $V$ unchanged.  Thus $D$ may
be taken to be the Schmidt rank and $p_i>0$ for every $i=1,\ldots,D$.

\section{Majorization and catalytic majorization}

Throughout this work, $x=(x_1,\ldots,x_D)$ and $y=(y_1,\ldots,y_D)$
are \emph{probability vectors}: their components are nonnegative and
\begin{equation}
 \sum_{i=1}^D x_i=\sum_{i=1}^D y_i=1.
 \label{supp:eq:normalization}
\end{equation}
If two probability vectors initially have different dimensions, zeros may be
appended to the shorter vector so that they are compared in a common
dimension.  We write $x^\downarrow$ for the vector obtained by arranging the
components of $x$ in nonincreasing order.

The majorization relation $x\prec y$ means
\begin{equation}
 \sum_{i=1}^k x_i^\downarrow
 \le \sum_{i=1}^k y_i^\downarrow
 \qquad (1\le k<D).
 \label{supp:eq:maj}
\end{equation}
At $k=D$ the two sums are equal one.
Nielsen's theorem states that a bipartite pure state with squared Schmidt
vector $x$ can be converted deterministically by LOCC into one with squared
Schmidt vector $y$ exactly when $x\prec y$ \cite{Nielsen1999}.  More
explicitly, in Schmidt bases,
\begin{equation*}
 \begin{aligned}
  \lvert\psi_x\rangle&=\sum_{i=1}^D\sqrt{x_i}\,
   \lvert i\rangle_A\lvert i\rangle_B,
  &
  \lvert\psi_y\rangle&=\sum_{i=1}^D\sqrt{y_i}\,
   \lvert i\rangle_A\lvert i\rangle_B,\\[-1mm]
  \lvert\psi_x\rangle&\xrightarrow{\mathrm{LOCC}}\lvert\psi_y\rangle
  &\Longleftrightarrow\quad x&\prec y.
 \end{aligned}
\end{equation*}

We write $x\prec_T y$ if there exists a finite-dimensional probability vector
$r$ such that
\begin{equation}
 x\otimes r\prec y\otimes r.
 \label{supp:eq:trumping-definition}
\end{equation}
For bipartite pure-state transformations this has a direct operational
meaning.  Let $\lvert\chi_r\rangle$ be a bipartite pure state whose squared
Schmidt vector is $r$.  By Nielsen's theorem, Eq.~\eqref{supp:eq:trumping-definition}
is equivalent to the existence of the deterministic LOCC transformation
\begin{equation*}
 \lvert\psi_x\rangle\otimes\lvert\chi_r\rangle
 \xrightarrow{\mathrm{LOCC}}
 \lvert\psi_y\rangle\otimes\lvert\chi_r\rangle.
\end{equation*}
Thus the auxiliary entangled state $\lvert\chi_r\rangle$ is a catalyst: the
same state occurs before and after the transformation.  The relation $\prec_T$ is also called \emph{trumping}
\cite{DaftuarKlimesh2001,Klimesh2007,Turgut2007}.

\subsection{Klimesh--Turgut characterization}

In the following, it will be useful to use an alternative characterization of trumping, obtained by Klimesh and Turgut \cite{Klimesh2007,Turgut2007}. For a probability vector
$x$, define the Klimesh function
\begin{equation}
 f_\alpha(x)=
 \begin{cases}
  \displaystyle \log\!\sum_i x_i^\alpha, & \alpha>1,\\[1mm]
  \displaystyle \sum_i x_i\log x_i=-H(x), & \alpha=1,\\[1mm]
  \displaystyle -\log\!\sum_i x_i^\alpha, & 0<\alpha<1,\\[1mm]
  \displaystyle -\sum_i\log x_i, & \alpha=0,\\[1mm]
  \displaystyle \log\!\sum_i x_i^\alpha, & \alpha<0.
 \end{cases}
 \label{supp:eq:Klimesh-functions}
\end{equation}
For $\alpha\le0$, if a vector has a zero component, the corresponding
Klimesh function is defined to be $+\infty$.

\begin{theorem}[Klimesh--Turgut~\cite{Klimesh2007,Turgut2007}]
\label{supp:thm:KT}
Let $x$ and $y$ be finite probability vectors, padded with zeros when
necessary so that they have the same dimension.  Suppose that they are not
permutations of one another, and suppose that $x$ and $y$ do not both contain
zero components.  Then
$x\prec_T y$ if and only if
\begin{equation}
 f_\alpha(x)<f_\alpha(y)
 \qquad\text{for every }\alpha\in\mathbb R.
 \label{supp:eq:KT-all}
\end{equation}
\end{theorem}

\subsection{Special case used below}

In the remainder of this Supplemental Material we need only a particularly
simple specialization of Theorem~\ref{supp:thm:KT}.  The source vector will
be
\begin{equation}
 x=p^{\otimes n},
\end{equation}
where $p$ has full support.  Hence every component of $x$ is strictly
positive.  Equivalently, the corresponding bipartite pure source has full
Schmidt rank on its $D^n$-dimensional Schmidt support.

The constructed output $y=s_n$ will have strictly fewer than $D^n$ nonzero
components.  We therefore pad $s_n$ with zeros to the same dimension as
$p^{\otimes n}$.  In this situation the conditions of nonpositive order are
automatic.  Indeed, because $x$ has full support,
\begin{equation}
 f_\alpha(x)<\infty
 \qquad (\alpha\le0),
\end{equation}
whereas the zero-padded output obeys
\begin{equation}
 f_\alpha(y)=+\infty
 \qquad (\alpha\le0).
\end{equation}
Thus
\begin{equation}
 f_\alpha(x)<f_\alpha(y)
 \qquad\text{for every }\alpha\le0.
\end{equation}
The same support difference also ensures
that $x$ and $y$ cannot be permutations of one another.

Consequently, catalytic
majorization in the present proof reduces to verifying only the conditions for $\alpha > 0$, i.e.,
\begin{align}
 \sum_i s_{n,i}^\alpha
 &>\left(\sum_i p_i^\alpha\right)^n,
 &&\alpha>1,
 \label{supp:eq:KT-special-high}\\
 H(s_n)&<nH(p),
 \label{supp:eq:KT-special-one}\\
 \sum_i s_{n,i}^\alpha
 &<\left(\sum_i p_i^\alpha\right)^n,
 &&0<\alpha<1.
 \label{supp:eq:KT-special-low}
\end{align}
Once Eqs.~\eqref{supp:eq:KT-special-high}--\eqref{supp:eq:KT-special-low}
are proved, Theorem~\ref{supp:thm:KT} guarantees a finite probability vector
$r_n$ such that
\begin{equation}
 p^{\otimes n}\otimes r_n\prec s_n\otimes r_n.
 \label{supp:eq:KT-special-conclusion}
\end{equation}
Nielsen's theorem then turns this majorization relation into the deterministic
LOCC transformation used in the main result.

\section{Construction of the target state}

Fix $0<\eta<1$ and a real number $\beta$.  Let
\begin{equation}
 m_n=\left\lfloor Rn+\beta\sqrt n\right\rfloor,
 \qquad
 N_n=2^{m_n},
 \label{supp:eq:mn}
\end{equation}
and define
\begin{equation}
 s_n=s_{N_n,\eta}:=
 \left(\eta,
 \underbrace{\frac{1-\eta}{N_n-1},\ldots,
 \frac{1-\eta}{N_n-1}}_{N_n-1\text{ entries}}
 \right).
 \label{supp:eq:sn}
\end{equation}
For all sufficiently large $n$, the first component is the largest.

\subsection{Fidelity, entropy, and support}

For any positive integer $N$, let
\begin{equation*}
 u_N:=\left(\frac1N,\ldots,\frac1N\right)
\end{equation*}
denote the uniform probability vector on $N$ components.  In particular, if
$N=2^m$, then $u_N=u_2^{\otimes m}$ is the squared Schmidt vector of $m$
Bell pairs.

\begin{lemma}
\label{supp:lem:fidelity}
For every fixed $0<\eta<1$,
\begin{equation}
 \Fid(s_{N,\eta},u_N)=
 \frac{\left[\sqrt\eta+\sqrt{(N-1)(1-\eta)}\right]^2}{N}
 \label{supp:eq:fid-exact}
\end{equation}
and
\begin{equation}
 \Fid(s_{N,\eta},u_N)>1-\eta
 \label{supp:eq:fid-bound}
\end{equation}
for all sufficiently large $N$.
\end{lemma}

\begin{proof}
Equation~\eqref{supp:eq:fid-exact} is direct.  Subtracting $1-\eta$ gives
\begin{equation}
 \Fid(s_{N,\eta},u_N)-(1-\eta)
 =\frac{2\eta-1+2\sqrt{\eta(1-\eta)(N-1)}}{N},
 \label{supp:eq:fid-diff}
\end{equation}
which is eventually positive.
\end{proof}

We use $h_2(q)=-q\log q-(1-q)\log(1-q)$ for the binary entropy.

\begin{lemma}
\label{supp:lem:entropy-support}
For fixed $0<\eta<1$ and fixed $\beta$,
\begin{equation}
 H(s_n)=h_2(\eta)+(1-\eta)\log(N_n-1)<nH
 \label{supp:eq:entropy-gap}
\end{equation}
for all sufficiently large $n$.  Moreover $N_n<D^n$ eventually, so $s_n$
padded to dimension $D^n$ has a zero component.
\end{lemma}

\begin{proof}
We first prove the entropy statement.  The vector $s_n$ has one component
$\eta$ and $N_n-1$ components equal to $(1-\eta)/(N_n-1)$.  Therefore
\begin{align*}
 H(s_n)
 &= -\eta\log\eta
    -(N_n-1)\frac{1-\eta}{N_n-1}
      \log\frac{1-\eta}{N_n-1} \\
 &= h_2(\eta)+(1-\eta)\log(N_n-1),
\end{align*}
which proves the identity in Eq.~\eqref{supp:eq:entropy-gap}.

Next, $N_n=2^{m_n}$ and
$m_n=Rn+\beta\sqrt n+O(1)$, where the $O(1)$ term accounts for the floor.
Because $N_n=2^{m_n}$ and $R=H$, we obtain
\[
 \log N_n=m_n
 =nH+\beta\sqrt n+O(1).
\]
Since $N_n$ grows exponentially with $n$, in particular $N_n\to\infty$.
Using base-2 logarithms,
\[
 \begin{aligned}
 \log(N_n-1)
 &=\log N_n+\log\!\left(1-\frac1{N_n}\right)\\
 &=\log N_n+O(N_n^{-1})
 =\log N_n+o(1),
 \end{aligned}
\]
where we used $\log(1-z)=O(z)$ as $z\to0$.  Therefore
\begin{equation}
 \log(N_n-1)=nH+\beta\sqrt n+O(1).
 \label{supp:eq:logN-asymptotic}
\end{equation}
Collecting these results we have
\[
 H(s_n)
 =(1-\eta)nH
 +(1-\eta)\beta\sqrt n
 +h_2(\eta)+O(1)
 =(1-\eta)nH+O(\sqrt n).
\]
The strict inequality in Eq.~\eqref{supp:eq:entropy-gap} is most transparent
after subtracting from the source entropy.  The preceding expansion gives
\[
 nH-H(s_n)=\eta nH+O(\sqrt n).
\]
Because $\eta>0$ and $H>0$ are fixed, the positive term $\eta nH$ grows
linearly in $n$ and therefore dominates the $O(\sqrt n)$ correction.  Thus
$H(s_n)<nH$ for all sufficiently large $n$.

For the support statement, $p$ has $D$ strictly positive components, so
$p^{\otimes n}$ has support size $D^n$.  Since $p$ is nonuniform,
$H<\log D$; set $\Delta=\log D-H>0$.  Then
\[
 \begin{split}
 n\log D-\log N_n
 &=n(\log D-H)-\beta\sqrt n+O(1)\\
 &=n\Delta+O(\sqrt n)>0
 \end{split}
\]
for all sufficiently large $n$.  Hence $N_n<D^n$.  After padding $s_n$ to
dimension $D^n$, it therefore has at least one zero component.
\end{proof}

\subsection{Uniform all-order power-sum inequalities}

Let
\begin{equation}
 g(\alpha)=\log\sum_{i=1}^D p_i^\alpha.
 \label{supp:eq:g}
\end{equation}
The function is analytic, with
\begin{equation}
 g(1)=0,
 \qquad
 g'(1)=-H,
 \qquad
 g''(1)=(\ln 2)V.
 \label{supp:eq:g-derivatives}
\end{equation}
Thus, as $x\to0$,
\begin{equation}
 g(1+x)=-Hx+\frac{(\ln 2)Vx^2}{2}+O(|x|^3).
 \label{supp:eq:cumulant}
\end{equation}
Because $p$ is nonuniform, its R\'enyi entropy
\begin{equation}
 H_\alpha(p)=\frac{g(\alpha)}{1-\alpha}
 \label{supp:eq:renyi}
\end{equation}
is strictly decreasing in $\alpha$ \cite{Renyi1961}.

\begin{lemma}[Orders below one]
\label{supp:lem:low}
Let
\begin{equation}
 c_\eta=-\log(1-\eta)>0,
 \qquad
 B_\eta=\sqrt{2(\ln 2)Vc_\eta}.
 \label{supp:eq:Beta}
\end{equation}
If $\beta<B_\eta$, then there exists a single integer $N$ such that, for
every $n\ge N$ and simultaneously for every order $0<\alpha<1$,
\begin{equation}
 \sum_i s_{n,i}^\alpha
 <\left(\sum_i p_i^\alpha\right)^n.
 \label{supp:eq:low-all}
\end{equation}
\end{lemma}

\begin{proof}
We prove the statement uniformly in $0<\alpha<1$ by contradiction.  If no
single threshold in $n$ worked for all such orders, there would exist a
sequence $n\to\infty$ and orders $\alpha_n\in(0,1)$ for which the desired
inequality fails.  Along this putative counterexample sequence,
\[
 \sum_i s_{n,i}^{\alpha_n}
 \ge \left(\sum_i p_i^{\alpha_n}\right)^n.
\]
Since both sides are positive, taking logarithms gives the equivalent
condition
\[
 \Delta_n:=n g(\alpha_n)-\log\sum_i s_{n,i}^{\alpha_n}\le0.
\]
We shall show that every possible asymptotic scale of $\alpha_n$ instead
forces $\Delta_n>0$ for all sufficiently large $n$.

Put
\[
 \delta_n=1-\alpha_n\in(0,1).
\]
The power sum of the constructed output vector in
Eq.~\eqref{supp:eq:sn} is
\[
 \sum_i s_{n,i}^{1-\delta_n}
 =\eta^{1-\delta_n}
 +(1-\eta)^{1-\delta_n}(N_n-1)^{\delta_n}.
\]

\emph{Orders bounded away from one.}
Suppose first that $\alpha_n$ is bounded away from one along a subsequence.
Equivalently, there is a fixed $\delta_0>0$ such that $\delta_n\ge\delta_0$.  Because
$p$ is nonuniform, $H_\alpha(p)>H$ for every $0\le\alpha<1$.  Continuity on
the compact interval $0\le\alpha\le1-\delta_0$ therefore gives a uniform
constant $g_0>0$ such that
\[
 H_{1-\delta_n}(p)-H\ge g_0.
\]
For the output, $0<\eta<1$ and $0<1-\delta_n<1$ imply
$\eta^{1-\delta_n}\le1$ and $(1-\eta)^{1-\delta_n}\le1$; moreover
$(N_n-1)^{\delta_n}\le N_n^{\delta_n}$.  Hence
\begin{equation}
 \sum_i s_{n,i}^{1-\delta_n}
 \le 1+N_n^{\delta_n}\le2N_n^{\delta_n}.
 \label{supp:eq:low-crude}
\end{equation}
Here the last inequality uses $N_n^{\delta_n}\ge1$.  Since
$g(1-\delta)=\delta H_{1-\delta}(p)$ and
\[
 \log N_n=m_n=nH+\beta\sqrt n+O(1),
\]
Eq.~\eqref{supp:eq:low-crude} gives
\begin{align}
 &n g(1-\delta_n)
 -\log\sum_i s_{n,i}^{1-\delta_n}
 \nonumber\\
 &\quad\ge
 n\delta_n\bigl[H_{1-\delta_n}(p)-H\bigr]
 -\beta\delta_n\sqrt n-O(1)>0
 \label{supp:eq:low-away}
\end{align}
for all sufficiently large $n$.  Indeed, the first term is at least
$n\delta_0g_0$, hence grows linearly in $n$, whereas the second is only
$O(\sqrt n)$.  Equation~\eqref{supp:eq:low-away} therefore gives
$\Delta_n>0$ for all sufficiently large $n$, while the assumed failure of
Eq.~\eqref{supp:eq:low-all} requires $\Delta_n\le0$.  These two conclusions
are incompatible.  Hence, for every fixed $\delta_0>0$, a violating sequence
can contain only finitely many indices with $\delta_n\ge\delta_0$, or
equivalently with $\alpha_n\le1-\delta_0$.  Any putative counterexample
sequence must therefore approach order one:
\[
 \delta_n\to0,
 \qquad\text{equivalently}\qquad
 \alpha_n\to1.
\]

\emph{The scale $t_n=\delta_n\sqrt n$.}
Set
\[
 t_n=\delta_n\sqrt n.
\]
Equation~\eqref{supp:eq:cumulant} with $x=-\delta$ gives
\[
 g(1-\delta)
 =H\delta+\frac{(\ln 2)V}{2}\delta^2+R(\delta),
 \qquad R(\delta)=O(\delta^3).
\]
Here $R(\delta)$ denotes the Taylor remainder.  Since
$R(\delta)/\delta^2\to0$ as $\delta\to0$ and $V>0$, there exists
$\delta_*>0$ such that, for $0<\delta<\delta_*$,
\[
 R(\delta)\ge-\frac{(\ln2)V}{4}\delta^2.
\]
Hence
\begin{equation}
 g(1-\delta)\ge H\delta+\frac{(\ln 2)V}{4}\delta^2.
 \label{supp:eq:strong-local}
\end{equation}
The weaker quadratic coefficient $(\ln 2)V/4$ is chosen only to turn the
Taylor expansion into a uniform lower bound; the quadratic Taylor coefficient
itself is $(\ln 2)V/2$.

If $t_n\to\infty$, Eq.~\eqref{supp:eq:strong-local} gives
\[
 n g(1-\delta_n)
 \ge nH\delta_n+\frac{(\ln 2)V}{4}t_n^2.
\]
On the other hand, Eq.~\eqref{supp:eq:low-crude} together with the
relation $\log N_n=m_n=nH+\beta\sqrt n+O(1)$ gives
\[
 \log\sum_i s_{n,i}^{1-\delta_n}
 \le nH\delta_n+\beta t_n+O(1).
\]
The common first-order term $nH\delta_n$ cancels, and hence
\begin{align}
 &n g(1-\delta_n)
 -\log\sum_i s_{n,i}^{1-\delta_n}
 \nonumber\\
 &\quad\ge
 \frac{(\ln 2)V}{4}t_n^2-\beta t_n-O(1)\longrightarrow+\infty.
 \label{supp:eq:low-t-infty}
\end{align}
The positive quadratic term dominates the linear one for every fixed
$\beta$, again contradicting $\Delta_n\le0$.

It remains to treat bounded $t_n$.  Passing to a subsequence, we may assume
$t_n\to t\ge0$.

\emph{Central scale: $t>0$.}
Suppose first that $t_n\to t>0$.  By definition of $t_n$,
$\delta_n=t_n/\sqrt n$.  Equation~\eqref{supp:eq:cumulant} then gives
\[
 n g(1-\delta_n)
 =Ht_n\sqrt n+\frac{(\ln 2)Vt_n^2}{2}+o(1),
\]
because $n\delta_n^3=t_n^3/\sqrt n\to0$.

For the output, write the power sum as
\[
 \sum_i s_{n,i}^{1-\delta_n}=A_n+B_n,
\]
where
\[
 A_n=\eta^{1-\delta_n},\qquad
 B_n=(1-\eta)^{1-\delta_n}(N_n-1)^{\delta_n}.
\]
Since $\delta_n\to0$, $A_n$ satisfies $A_n\to\eta$, so it stays of order one.  In contrast,
using Eq.~\eqref{supp:eq:logN-asymptotic},
\[
 \delta_n\log(N_n-1)
 =Ht_n\sqrt n+\beta t_n+o(1)\longrightarrow+\infty.
\]
Thus $B_n\to\infty$ and $A_n/B_n\to0$. We further obtain
\[
 \begin{aligned}
 \log(A_n+B_n)
 &=\log B_n+\log\!\left(1+\frac{A_n}{B_n}\right)\\
 &=\log B_n+o(1).
 \end{aligned}
\]
Since $c_\eta=-\log(1-\eta)$,
\[
 \log B_n
 =Ht_n\sqrt n+\beta t_n-c_\eta+o(1).
\]

Subtracting the source and output expansions cancels the entire
$Ht_n\sqrt n$ term and yields
\begin{align}
 &n g(1-\delta_n)
 -\log\sum_i s_{n,i}^{1-\delta_n}
 \nonumber\\
 &\quad\longrightarrow
 Q_\beta(t):=\frac{(\ln 2)Vt^2}{2}-\beta t+c_\eta.
 \label{supp:eq:Qlimit}
\end{align}

For $\beta\le0$, every term is nonnegative, so $Q_\beta(t)>0$.  For $\beta>0$, the quadratic is minimized
at $t_*=\beta/[(\ln 2)V]$, and
\begin{equation}
 \min_{t\ge0}Q_\beta(t)
 =c_\eta-\frac{\beta^2}{2(\ln 2)V}>0
 \label{supp:eq:Qmin}
\end{equation}
by the assumption $\beta<B_\eta$.  Hence the logarithmic difference is
strictly positive for all sufficiently large $n$, contradicting the assumed
failure also on the central $t>0$ scale.

\emph{Scale $t_n\to0$.}
We are left with $t_n\to0$, i.e. $\delta_n=o(n^{-1/2})$.  At this resolution
the variable $t_n$ no longer distinguishes all possibilities, so introduce
\[
 \lambda_n=n\delta_n.
\]

If $\lambda_n\to\infty$, since we are in the regime $t_n=\delta_n\sqrt n\to0$, we have
$n\delta_n^2=t_n^2\to0$.  Also
$n\delta_n^3=t_n^2\delta_n\to0$.  Substituting $x=-\delta_n$ into
Eq.~\eqref{supp:eq:cumulant} therefore gives
\[
 \begin{aligned}
 n g(1-\delta_n)
 &=nH\delta_n+\frac{(\ln 2)V}{2}n\delta_n^2
   +O(n\delta_n^3)\\
 &=\lambda_nH+o(1).
 \end{aligned}
\]
To obtain the corresponding output expansion, write again
\[
 A_n=\eta^{1-\delta_n},\qquad
 B_n=(1-\eta)^{1-\delta_n}(N_n-1)^{\delta_n}.
\]
Using Eq.~\eqref{supp:eq:logN-asymptotic}, we find
\[
 \delta_n\log(N_n-1)
 =\lambda_nH+\beta t_n+O(\delta_n)
 =\lambda_nH+o(1).
\]
Therefore
\[
 \begin{aligned}
 \log B_n
 &=(1-\delta_n)\log(1-\eta)
   +\delta_n\log(N_n-1)\\
 &=\lambda_nH-c_\eta+o(1).
 \end{aligned}
\]
Because $\lambda_n\to\infty$ and $H>0$, this implies $B_n\to\infty$,
whereas $A_n\to\eta$.  Hence $A_n/B_n\to0$, and we obtain
\[
 \begin{aligned}
 \log(A_n+B_n)
 &=\log B_n+\log\!\left(1+\frac{A_n}{B_n}\right)\\
 &=\log B_n+o(1)\\
 &=\lambda_nH-c_\eta+o(1).
 \end{aligned}
\]
Thus
\begin{equation}
 n g(1-\delta_n)
 -\log\sum_i s_{n,i}^{1-\delta_n}
 \longrightarrow c_\eta>0.
 \label{supp:eq:low-lambda-infty}
\end{equation}
This is again incompatible with $\Delta_n\le0$.

Finally suppose $\lambda_n$ is bounded.  Passing to a further subsequence,
let $\lambda_n\to\lambda\ge0$.  If $\lambda>0$, then, recalling that
$\delta_n=\lambda_n/n$ and using
Eq.~\eqref{supp:eq:cumulant} gives
\[
 \begin{aligned}
 n g(1-\delta_n)
 &=H\lambda_n
   +\frac{(\ln2)V}{2}\frac{\lambda_n^2}{n}
   +O\!\left(\frac{\lambda_n^3}{n^2}\right)\\
 &\longrightarrow\lambda H,
 \end{aligned}
\]
because $\lambda_n$ is bounded. Moreover, multiplying
Eq.~\eqref{supp:eq:logN-asymptotic} by
$\delta_n=\lambda_n/n$ gives
\[
 \begin{aligned}
 \delta_n\log(N_n-1)
 &=\frac{\lambda_n}{n}
   \left(nH+\beta\sqrt n+O(1)\right)\\
 &=\lambda_nH+\frac{\beta\lambda_n}{\sqrt n}
   +O\!\left(\frac{\lambda_n}{n}\right)
 \longrightarrow\lambda H.
 \end{aligned}
\]
Thus $(N_n-1)^{\delta_n}\to2^{\lambda H}$, while
$(1-\eta)^{1-\delta_n}\to1-\eta$, and therefore
\[
 (1-\eta)^{1-\delta_n}(N_n-1)^{\delta_n}
 \longrightarrow(1-\eta)2^{\lambda H}.
\]
Consequently,
\begin{align}
 n g(1-\delta_n)&\longrightarrow\lambda H,\\
 \sum_i s_{n,i}^{1-\delta_n}
 &\longrightarrow \eta+(1-\eta)2^{\lambda H}.
 \label{supp:eq:low-lambda-limit}
\end{align}
The logarithmic difference therefore tends to
\begin{equation}
 \lambda H-\log\left[\eta+(1-\eta)2^{\lambda H}\right]>0.
 \label{supp:eq:low-lambda-positive}
\end{equation}
To see this, note that  $\lambda H>0$,
$2^{\lambda H}>1$, and therefore
\[
 \eta+(1-\eta)2^{\lambda H}<2^{\lambda H};
\]
taking logarithms gives Eq.~\eqref{supp:eq:low-lambda-positive}.

It remains to consider the case $\lambda_n\to0$. Since
$\delta_n=\lambda_n/n$, Eq.~\eqref{supp:eq:cumulant} gives
\[
 \begin{aligned}
 n g(1-\delta_n)
 &=H\lambda_n
   +\frac{(\ln2)V}{2}\frac{\lambda_n^2}{n}
   +O\!\left(\frac{\lambda_n^3}{n^2}\right)\\
 &=H\lambda_n+o(\lambda_n),
 \end{aligned}
\]
because $(\lambda_n^2/n)/\lambda_n=\lambda_n/n\to0$ and likewise for
the cubic term.

For the output, note first that $\delta_n=\lambda_n/n=o(\lambda_n)$.
For any fixed $a\in(0,1)$, use the Taylor expansion
$2^x=e^{(\ln2)x}=1+(\ln2)x+O(x^2)$ as $x\to0$.  With
$x=-\delta_n\log a$, this gives
\[
 a^{1-\delta_n}
 =a\,2^{-\delta_n\log a}
 =a\left[1-(\ln2)(\log a)\delta_n+O(\delta_n^2)\right]
 =a+o(\lambda_n).
\]
Here the last equality uses $\delta_n=o(\lambda_n)$.
Applying this result gives
\[
 \eta^{1-\delta_n}=\eta+o(\lambda_n),\qquad
 (1-\eta)^{1-\delta_n}=(1-\eta)+o(\lambda_n).
\]
Moreover, multiplying Eq.~\eqref{supp:eq:logN-asymptotic} by
$\delta_n=\lambda_n/n$ gives
\[
 \begin{aligned}
 \delta_n\log(N_n-1)
 &=H\lambda_n+\frac{\beta\lambda_n}{\sqrt n}
   +O\!\left(\frac{\lambda_n}{n}\right)\\
 &=H\lambda_n+o(\lambda_n).
 \end{aligned}
\]
If $x_n=H\lambda_n+o(\lambda_n)$, then $x_n\to0$ and
$2^{x_n}=1+(\ln2)x_n+o(x_n)$.  Therefore
\[
 (N_n-1)^{\delta_n}
 =2^{H\lambda_n+o(\lambda_n)}
 =1+(\ln2)H\lambda_n+o(\lambda_n).
\]
By similar arguments we find
\[
 \begin{aligned}
 &(1-\eta)^{1-\delta_n}(N_n-1)^{\delta_n}\\
 &\quad=\bigl[(1-\eta)+o(\lambda_n)\bigr]
 \bigl[1+(\ln2)H\lambda_n+o(\lambda_n)\bigr]\\
 &\quad=(1-\eta)+(1-\eta)(\ln2)H\lambda_n
 +o(\lambda_n).
 \end{aligned}
\]
All cross terms are absorbed into $o(\lambda_n)$.  Indeed, if
$u_n=o(\lambda_n)$ and $v_n=o(\lambda_n)$ denote the two remainder terms,
then constant multiples of $u_n$ and $v_n$ are $o(\lambda_n)$, while
$\lambda_n u_n$, $\lambda_n v_n$, and $u_nv_n$ are also $o(\lambda_n)$
because $\lambda_n\to0$.

Adding the contribution $\eta+o(\lambda_n)$ yields
\[
 \sum_i s_{n,i}^{1-\delta_n}
 =1+(1-\eta)(\ln2)H\lambda_n+o(\lambda_n).
\]
Finally, using $\log_2(1+x)=x/\ln2+o(x)$,
\begin{equation}
 \log\sum_i s_{n,i}^{1-\delta_n}
 =(1-\eta)H\lambda_n+o(\lambda_n).
 \label{supp:eq:low-lambda-zero}
\end{equation}
Subtracting Eq.~\eqref{supp:eq:low-lambda-zero} from the source expansion gives
\[
 \Delta_n
 =\eta H\lambda_n+o(\lambda_n).
\]
Write the remainder as $r_n=o(\lambda_n)$.  Since $\eta H>0$ and
$r_n/\lambda_n\to0$, for all sufficiently large $n$,
$|r_n|\le(\eta H/2)\lambda_n$.  Because $\lambda_n>0$, it follows that
\[
 \Delta_n\ge\frac{\eta H}{2}\lambda_n>0.
\]
This final possibility also contradicts the assumed failure.

We have exhausted every possible asymptotic scale of a putative
counterexample sequence. Every case forces $\Delta_n>0$, contrary to the defining
counterexample condition $\Delta_n\le0$.  Therefore no such sequence exists,
and Eq.~\eqref{supp:eq:low-all} holds uniformly for all $0<\alpha<1$ once
$n$ is sufficiently large.
\end{proof}

\begin{lemma}[Orders above one]
\label{supp:lem:high}
For every fixed $0<\eta<1$ and every fixed real $\beta$, there exists an
integer $N$ such that for every $n\ge N$ and every $\alpha>1$,
\begin{equation}
 \sum_i s_{n,i}^\alpha
 >\left(\sum_i p_i^\alpha\right)^n.
 \label{supp:eq:high-all}
\end{equation}
In particular, the threshold $N$ is uniform over the whole interval
$\alpha>1$.
\end{lemma}

\begin{proof}
We prove the uniform statement by contradiction.  Suppose that no single
threshold $N$ works for all $\alpha>1$.  Then there exist a sequence
$n\to\infty$ and orders $\alpha_n>1$ for which Eq.~\eqref{supp:eq:high-all}
fails, i.e.
\[
 \sum_i s_{n,i}^{\alpha_n}
 \le \left(\sum_i p_i^{\alpha_n}\right)^n.
\]
Put
\[
 \delta_n=\alpha_n-1>0.
\]
We show that every possible asymptotic behavior of $\delta_n$ contradicts
this assumed failure.

First suppose that the violating orders remain a fixed distance above one.
Precisely, suppose that along a subsequence there is a constant
$\delta_0>0$ such that
\[
 \delta_n\ge\delta_0,
 \qquad\text{equivalently}\qquad
 \alpha_n\ge1+\delta_0.
\]
Let $q=\|p\|_{1+\delta_0}<1$.  Here
$\|x\|_r=(\sum_i x_i^r)^{1/r}$ denotes the $\ell_r$ norm.  For a
fixed nonnegative finite vector these norms are nonincreasing in $r$
\cite{HardyLittlewoodPolya1952}, hence
\[
 \|p\|_{\alpha_n}\le\|p\|_{1+\delta_0}=q.
\]
Moreover $q<1$: because $p$ is a nontrivial probability vector with full
support, for every $r>1$ one has $p_i^r<p_i$ for every component with
$0<p_i<1$, and therefore $\sum_i p_i^r<\sum_i p_i=1$.
Tensor products factorize the $\ell_r$ norm, while the $\ell_r$ norm of a
vector is at least as large as any one of its components.  Since $s_n$
contains the fixed component $\eta$, we obtain
\begin{equation}
 \left(\sum_i p_i^{\alpha_n}\right)^{n/\alpha_n}
 =\|p\|_{\alpha_n}^n\le q^n,
 \qquad
 \left(\sum_i s_{n,i}^{\alpha_n}\right)^{1/\alpha_n}\ge\eta.
 \label{supp:eq:high-away}
\end{equation}
Because $0<q<1$, $q^n\to0$, so $q^n<\eta$ for all sufficiently large $n$.
Equation~\eqref{supp:eq:high-away} then gives
\[
 \left(\sum_i p_i^{\alpha_n}\right)^{n/\alpha_n}
 <
 \left(\sum_i s_{n,i}^{\alpha_n}\right)^{1/\alpha_n}.
\]
Raising both sides to the positive power $\alpha_n$ yields
\[
 \left(\sum_i p_i^{\alpha_n}\right)^n
 <\sum_i s_{n,i}^{\alpha_n},
\]
which is the desired high-order inequality and contradicts the assumed
failure.  Thus, every putative counterexample sequence must satisfy
\[
 \delta_n=\alpha_n-1\longrightarrow0.
\]

We may therefore assume $\delta_n\to0$ and put
\[
 \lambda_n=n\delta_n.
\]
Consider first the case $\lambda_n\to\infty$.  Substituting
$x=\delta_n$ into the expansion \eqref{supp:eq:cumulant} and
multiplying by $n$ gives
\[
 n g(1+\delta_n)
 =-nH\delta_n+\frac{(\ln2)V}{2}n\delta_n^2
   +O(n\delta_n^3).
\]
Since $n\delta_n=\lambda_n$ and $\delta_n\to0$,
\[
 n\delta_n^2=\lambda_n\delta_n=o(\lambda_n),
 \qquad
 n\delta_n^3=\lambda_n\delta_n^2=o(\lambda_n).
\]
Hence
\begin{equation}
 n g(1+\delta_n)=-\lambda_nH+o(\lambda_n)
 \longrightarrow-\infty,
 \label{supp:eq:high-lambda-infty}
\end{equation}
because $H>0$ and $\lambda_n\to\infty$.  The source power sum is therefore
\[
 \left(\sum_i p_i^{\alpha_n}\right)^n
 =2^{\,n g(1+\delta_n)}\longrightarrow0.
\]
On the other hand, the output power sum fulfills
\[
 \sum_i s_{n,i}^{\alpha_n}
 \ge \eta^{1+\delta_n}\longrightarrow\eta>0.
\]
Consequently, for all sufficiently large $n$,
\[
 \left(\sum_i p_i^{\alpha_n}\right)^n
 <\eta^{1+\delta_n}
 \le\sum_i s_{n,i}^{\alpha_n},
\]
again contradicting the assumed failure.

It remains to consider bounded $\lambda_n$.  Passing to a subsequence, we may
assume
\[
 \lambda_n\longrightarrow\lambda\ge0.
\]
Because $\delta_n=\lambda_n/n$, the expansion in Eq.~\eqref{supp:eq:cumulant} gives
\[
 n g(1+\delta_n)
 =-H\lambda_n+O(\lambda_n^2/n),
\]
so if $\lambda>0$,
\[
 \left(\sum_i p_i^{1+\delta_n}\right)^n
 =2^{\,n g(1+\delta_n)}
 \longrightarrow2^{-\lambda H}.
\]
For the output, use its exact power sum
\begin{equation}
 \sum_i s_{n,i}^{1+\delta_n}
 =\eta^{1+\delta_n}
 +(1-\eta)^{1+\delta_n}(N_n-1)^{-\delta_n}.
 \label{supp:eq:high-output-exact}
\end{equation}
Since $\delta_n\to0$,
\[
 \eta^{1+\delta_n}\longrightarrow\eta,
 \qquad
 (1-\eta)^{1+\delta_n}\longrightarrow1-\eta.
\]
Furthermore,
\[
 \delta_n\log(N_n-1)
 =\frac{\lambda_n}{n}\bigl(nH+\beta\sqrt n+O(1)\bigr)
 \longrightarrow\lambda H,
\]
so
\[
 (N_n-1)^{-\delta_n}\longrightarrow2^{-\lambda H}.
\]
Thus
\begin{align}
 \left(\sum_i p_i^{1+\delta_n}\right)^n
 &\longrightarrow 2^{-\lambda H},\\
 \sum_i s_{n,i}^{1+\delta_n}
 &\longrightarrow \eta+(1-\eta)2^{-\lambda H}>2^{-\lambda H}.
 \label{supp:eq:high-lambda-limit}
\end{align}
The strict inequality follows from
\[
 \eta+(1-\eta)2^{-\lambda H}-2^{-\lambda H}
 =\eta\bigl(1-2^{-\lambda H}\bigr)>0,
\]
because $\eta>0$, $H>0$, and $\lambda>0$.  Hence the output power sum is
eventually strictly larger than the source power sum, contradicting the
assumed failure.

Finally suppose $\lambda_n\to0$.  Since
$\delta_n=\lambda_n/n$, substituting $x=\delta_n$ into
Eq.~\eqref{supp:eq:cumulant} and multiplying by $n$ gives
\[
 \begin{aligned}
 n g(1+\delta_n)
 &=-H\lambda_n
   +\frac{(\ln2)V}{2}\frac{\lambda_n^2}{n}
   +O\!\left(\frac{\lambda_n^3}{n^2}\right)\\
 &=-H\lambda_n+o(\lambda_n),
 \end{aligned}
\]
because $(\lambda_n^2/n)/\lambda_n=\lambda_n/n\to0$, and the cubic
term is smaller still.

For the output, start from the exact identity
Eq.~\eqref{supp:eq:high-output-exact}.  Since
$\delta_n=\lambda_n/n=o(\lambda_n)$, for every fixed $a\in(0,1)$ the
Taylor expansion $2^x=1+(\ln2)x+O(x^2)$ gives
\[
 a^{1+\delta_n}
 =a\,2^{\delta_n\log a}
 =a\left[1+(\ln2)(\log a)\delta_n+O(\delta_n^2)\right]
 =a+o(\lambda_n).
\]
Thus
\[
 \eta^{1+\delta_n}=\eta+o(\lambda_n),\qquad
 (1-\eta)^{1+\delta_n}=(1-\eta)+o(\lambda_n).
\]
Moreover, multiplying Eq.~\eqref{supp:eq:logN-asymptotic} by
$\delta_n=\lambda_n/n$ gives
\[
 \begin{aligned}
 \delta_n\log(N_n-1)
 &=H\lambda_n+\frac{\beta\lambda_n}{\sqrt n}
   +O\!\left(\frac{\lambda_n}{n}\right)\\
 &=H\lambda_n+o(\lambda_n).
 \end{aligned}
\]
Using again $2^x=1+(\ln2)x+o(x)$ at $x=0$ therefore yields
\[
 (N_n-1)^{-\delta_n}
 =2^{-H\lambda_n+o(\lambda_n)}
 =1-(\ln2)H\lambda_n+o(\lambda_n).
\]
Collecting these results we obtain
\[
 \begin{aligned}
 &(1-\eta)^{1+\delta_n}(N_n-1)^{-\delta_n}\\
 &\quad=\bigl[(1-\eta)+o(\lambda_n)\bigr]
 \bigl[1-(\ln2)H\lambda_n+o(\lambda_n)\bigr]\\
 &\quad=(1-\eta)-(1-\eta)(\ln2)H\lambda_n+o(\lambda_n).
 \end{aligned}
\]
For the output power sum this implies, together with Eq. (\ref{supp:eq:high-output-exact}):
\[
 \sum_i s_{n,i}^{1+\delta_n}
 =1-(1-\eta)(\ln2)H\lambda_n+o(\lambda_n).
\]
Using $\log_2(1+x)=x/\ln2+o(x)$ as $x\to0$ gives
\begin{equation}
 \log\sum_i s_{n,i}^{1+\delta_n}
 =-(1-\eta)\lambda_nH+o(\lambda_n).
 \label{supp:eq:high-lambda-zero}
\end{equation}
Subtracting the source expansion from
Eq.~\eqref{supp:eq:high-lambda-zero} yields
\[
 \log\sum_i s_{n,i}^{1+\delta_n}
 -n g(1+\delta_n)
 =\eta H\lambda_n+o(\lambda_n).
\]
This expression is positive for all sufficiently large $n$. Thus the output power sum is again strictly larger than the source power
sum, contradicting the assumed failure.

We have exhausted every possible asymptotic scale of a putative violating
sequence.  Therefore no such sequence exists, and a single threshold $N$
works simultaneously for every $\alpha>1$.
\end{proof}

\section{Catalytic achievability}

We now combine the preceding results to establish catalytic achievability.

\begin{proposition}
\label{supp:prop:trumping}
If $\beta<B_\eta$, then
\begin{equation}
 p^{\otimes n}\prec_T s_n
 \label{supp:eq:trumping}
\end{equation}
for every sufficiently large $n$.
\end{proposition}

\begin{proof}
Lemma~\ref{supp:lem:low} proves all $0<\alpha<1$ inequalities,
Lemma~\ref{supp:lem:entropy-support} gives strict entropy decrease, and
Lemma~\ref{supp:lem:high} proves all $\alpha>1$ inequalities.  The padded
$s_n$ has a zero component while $p^{\otimes n}$ has full support, so all
nonpositive-order conditions are automatic.  The Klimesh--Turgut theorem,
Theorem~\ref{supp:thm:KT}, therefore gives Eq.~\eqref{supp:eq:trumping} with
a finite catalyst.
\end{proof}

\begin{proposition}[Catalytic lower bound]
\label{supp:prop:achievability}
For every $0<\eps<1$,
\begin{equation}
 \liminf_{n\to\infty}
 \frac{M_n^{\cat}(\eps)-nH}{\sqrt n}
 \ge \sqrt{2(\ln 2)V\log[1/(1-\eps^2)]}.
 \label{supp:eq:achievability}
\end{equation}
\end{proposition}

\begin{proof}
Set $\eta=\eps^2$ and fix any
\begin{equation}
 0<\beta<B_\eta = \sqrt{2(\ln 2)V[-\log(1-\eta)]}.
\end{equation}
By Proposition~\ref{supp:prop:trumping}, a finite exact-return catalyst exists
for $p^{\otimes n}\to s_n$ for every sufficiently large $n$.  By
Lemma~\ref{supp:lem:fidelity},
\begin{equation}
 \Fid(s_n,u_2^{\otimes m_n})>1-\eta=1-\eps^2
\end{equation}
eventually.  Hence $M_n^{\cat}(\eps)\ge m_n$ and therefore
\[
 \liminf_{n\to\infty}
 \frac{M_n^{\cat}(\eps)-nH}{\sqrt n}\ge\beta.
\]
Letting $\beta$ increase to the displayed endpoint proves
Eq.~\eqref{supp:eq:achievability}.
\end{proof}

\section{Catalytic upper bound}

We next prove that no exact-return catalyst can produce a larger
second-order coefficient.  The converse applies to an arbitrary feasible
catalytic protocol, without assuming a specific form of the achieving
output. The following lemma converts
the target-fidelity constraint into the complementary lower bound on the
output power sum.

\begin{lemma}
\label{supp:lem:Jensen}
Let $s$ be any probability vector and let
$u_N=(1/N,\ldots,1/N)$ be uniform on $N$ coordinates.  If
\begin{equation}
 \Fid(s,u_N)\ge1-\eps^2,
 \label{supp:eq:fid-assume}
\end{equation}
then for every $\alpha\ge1/2$,
\begin{equation}
 \sum_i s_i^\alpha
 \ge N^{1-\alpha}(1-\eps^2)^\alpha.
 \label{supp:eq:Jensen-bound}
\end{equation}
\end{lemma}

\begin{proof}
Let $I$ be the $N$ coordinates on which $u_N$ is supported and put
$x_i=\sqrt{Ns_i}$ for $i\in I$.  The fidelity assumption is equivalent to
\begin{equation}
 \frac1N\sum_{i\in I}x_i\ge\sqrt{1-\eps^2}.
\end{equation}
Since $2\alpha\ge1$, the function $z\mapsto z^{2\alpha}$ is convex on
$[0,\infty)$.  Moreover, from $x_i=\sqrt{Ns_i}$ we have
$s_i^\alpha=N^{-\alpha}x_i^{2\alpha}$ for $i\in I$.  Therefore
\begin{equation}
\begin{aligned}
 \sum_i s_i^\alpha
 &\ge \sum_{i\in I}s_i^\alpha
 =N^{-\alpha}\sum_{i\in I}x_i^{2\alpha}\\
 &=N^{1-\alpha}
 \left(\frac1N\sum_{i\in I}x_i^{2\alpha}\right)\\
 &\ge N^{1-\alpha}
 \left(\frac1N\sum_{i\in I}x_i\right)^{2\alpha}\\
 &\ge N^{1-\alpha}(1-\eps^2)^\alpha.
\end{aligned}
\end{equation}
The first inequality restricts the sum to the $N$ target-support coordinates.
The second follows from Jensen's inequality for the convex function
$z^{2\alpha}$, and the last uses the fidelity assumption above.
\end{proof}

\begin{proposition}[Catalytic upper bound]
\label{supp:prop:converse}
For every $0<\eps<1$,
\begin{equation}
 \limsup_{n\to\infty}
 \frac{M_n^{\cat}(\eps)-nH}{\sqrt n}
 \le \sqrt{2(\ln 2)V\log[1/(1-\eps^2)]}.
 \label{supp:eq:converse}
\end{equation}
\end{proposition}

\begin{proof}
For each $n$, let $m_n$ be any catalytically feasible output-copy number, with
output $s_n$ and catalyst $r_n$.  Fix $t>0$ and set
\begin{equation}
 \alpha_n=1-\frac{t}{\sqrt n}.
\end{equation}
For all sufficiently large $n$, $1/2<\alpha_n<1$.  Feasibility means
$p^{\otimes n}\otimes r_n\prec s_n\otimes r_n$.  For
$0<\alpha_n<1$, the function $x\mapsto x^{\alpha_n}$ is concave, so
Schur concavity of the power sum gives
\begin{equation}
 \sum_{i,j}(p_i^{\otimes n}r_{n,j})^{\alpha_n}
 \ge\sum_{i,j}(s_{n,i}r_{n,j})^{\alpha_n}.
\end{equation}
The positive factor $\sum_jr_{n,j}^{\alpha_n}$ cancels, so
\begin{equation}
 \left(\sum_i p_i^{\alpha_n}\right)^n
 \ge\sum_i s_{n,i}^{\alpha_n}.
 \label{supp:eq:catalyst-cancel}
\end{equation}
Together with
Lemma~\ref{supp:lem:Jensen} we obtain
\begin{equation}
 \left(\sum_i p_i^{\alpha_n}\right)^n
 \ge \sum_i s_{n,i}^{\alpha_n}
 \ge N_n^{1-\alpha_n}(1-\eps^2)^{\alpha_n}.
 \label{supp:eq:converse-chain}
\end{equation}
Taking base-2 logarithms, using
$g(\alpha)=\log\sum_i p_i^\alpha$ and $\log N_n=m_n$, yields
\begin{equation}
 n g(\alpha_n)
 \ge (1-\alpha_n)m_n
 +\alpha_n\log(1-\eps^2).
 \label{supp:eq:converse-log}
\end{equation}
Write
\begin{equation}
 b_n=\frac{m_n-Rn}{\sqrt n}.
\end{equation}
Since $\alpha_n=1-t/\sqrt n$, Eq.~\eqref{supp:eq:cumulant} gives
\begin{equation}
 n g(\alpha_n)
 =Ht\sqrt n+\frac{(\ln2)Vt^2}{2}+O(n^{-1/2}).
 \label{supp:eq:converse-source-expansion}
\end{equation}
On the other hand, from $m_n=Rn+b_n\sqrt n$ and $R=H$,
\begin{align}
 &(1-\alpha_n)m_n+\alpha_n\log(1-\eps^2)\nonumber\\
 &\quad=Ht\sqrt n+t b_n+\log(1-\eps^2)+O(n^{-1/2}).
 \label{supp:eq:converse-rhs-expansion}
\end{align}
Substituting Eqs.~\eqref{supp:eq:converse-source-expansion} and
\eqref{supp:eq:converse-rhs-expansion} into
Eq.~\eqref{supp:eq:converse-log}, the common term $Ht\sqrt n$ cancels.  Thus
\[
 \frac{(\ln2)Vt^2}{2}
 \ge t b_n+\log(1-\eps^2)+o(1).
\]
Dividing by $t>0$ and rearranging gives
\begin{equation}
 b_n\le
 \left[
 \frac{(\ln 2)Vt}{2}+\frac{-\log(1-\eps^2)}{t}\right]+o(1).
 \label{supp:eq:converse-bn}
\end{equation}
Thus the limsup of every feasible sequence obeys the same bound.  Minimizing
over $t>0$ at
\begin{equation}
 t_*=\sqrt{\frac{2[-\log(1-\eps^2)]}{(\ln 2)V}}
\end{equation}
gives Eq.~\eqref{supp:eq:converse}.
\end{proof}

Theorem~\ref{thm:main} follows directly from Propositions~\ref{supp:prop:achievability} and
\ref{supp:prop:converse}.

\end{document}